\documentclass[12pt,a4paper,reqno]{amsart}
\usepackage{amsfonts}
\usepackage{amsthm}
\usepackage{amsmath}
\usepackage{times}
\usepackage{amscd}
\usepackage[latin2]{inputenc}
\usepackage{t1enc}
\usepackage{enumitem}
\usepackage[mathscr]{eucal}
\usepackage{indentfirst}
\usepackage{graphicx}
\usepackage{graphics}
\usepackage{pict2e}
\usepackage{epic}
\usepackage{esint}
\usepackage[margin=2.9cm]{geometry}
\usepackage{epstopdf} 
\usepackage{verbatim}
\usepackage{cancel}
\usepackage{amssymb}
\usepackage{xcolor}
\usepackage{dsfont}
\usepackage{hyperref}
\usepackage{physics}
\usepackage{cancel}

\setitemize{itemsep=5pt, topsep=7pt}
\numberwithin{equation}{section}

\theoremstyle{plain}
\newtheorem{theo}{Theorem}[section]
\newtheorem{lem}{Lemma}[section]

\newtheorem{prop}{Proposition}[section]
\newtheorem{Conj}{Conjecture}[section]

\theoremstyle{definition}

\newtheorem{rem}{Remark}[section]

\newcommand{\unit}{1\!\!1}
\newcommand{\R}{\mathbb{R}}
\newcommand{\Rtre}{\mathbb{R}^{3}}

\newcommand{\minG}{\mathcal{M}_{\mathcal{G}}}
\newcommand{\G}{\mathcal{G}}

\newcommand{\E}{\mathcal{E}}
\newcommand{\F}{\mathcal{F}}
\newcommand{\minE}{\mathcal{M}_{\mathcal{E}}}
\newcommand{\minF}{\mathcal{M}_{\mathcal{F}}}
\newcommand{\Pel}{\text{X}_{\text{el}}}
\newcommand{\Vel}{\text{V}_{\text{el}}}
\newcommand{\Pph}{\text{X}_{\text{ph}}}
\newcommand{\Vph}{\text{V}_{\text{ph}}}

\DeclareMathOperator{\spn}{span}
\DeclareMathOperator{\infspec}{inf\,spec}

\DeclareMathOperator{\dist}{dist}

\begin{document}

\title[]{The Effective Mass Problem for the Landau--Pekar Equations}

\author{Dario Feliciangeli,  Simone Rademacher and Robert Seiringer}
\address{IST Austria, Am Campus 1, 3400 Klosterneuburg, Austria}

\date{\today}

\maketitle

\begin{abstract}
We provide a definition of the effective mass for the classical polaron described by the Landau--Pekar equations. It is based on a novel variational principle, minimizing the energy functional over states with given (initial) velocity. The resulting formula for the polaron's effective mass agrees with the prediction by Landau and Pekar \cite{landau1948effective}. 
\end{abstract}

\section{Introduction} 

The polaron is a model of an electron interacting with its self-induced polarization field of the underlying crystal. The description of the polarization as a quantum field corresponds to the Fr\"ohlich model \cite{frohlich}.  In the  classical approximation, on the other hand, the dynamics of a polaron is described by the Landau--Pekar (LP) equations \cite{Landau,Pekar,landau1948effective}.  For $(\psi_t, \varphi_t) \in H^1(\Rtre)\times L^2(\Rtre)$, where $\psi_t$ is the electron wave function and $\varphi_t$  denotes the phonon field, these equations read in suitable units (see \cite{frohlich} or \cite{alexandrov2010advances})
\begin{equation}\label{eq:LP}
\begin{aligned}
i \partial_t \psi_t  & = h_{ \varphi_t  }  \psi_t,  \\
i \alpha^2 \partial_t \varphi_t  & = \varphi_t + \sigma_{\psi_t},
\end{aligned}
\end{equation}
where $h_{\varphi}$ is the Schr\"odinger operator 
\begin{equation}
h_{\varphi} = - \Delta  + V_\varphi
\end{equation}
with potential 
\begin{equation}
V_\varphi(x) = 2 
\Re\, [ ( - \Delta)^{-1/2} \varphi ] (x)= 
\pi^{-2} |x|^{-2} * \Re\, \varphi ,
\end{equation}
and
\begin{equation}
\sigma_{\psi} (x) = 
\left[ (- \Delta)^{-1/2} \vert \psi \vert^2\right]  (x) = 
(2\pi^2)^{-1}  |x|^{-2} * |\psi|^2 ,
\end{equation}
where $*$ denotes convolution. 
The parameter $\alpha>0$ quantifies  the strength of the coupling of the electron's charge to the polarization field. 

Despite its long history, the polaron model continues being actively investigated. For recent experimental and numerical work, we refer to \cite{Franchini,Mishchenko,Puppin,sio,sio2} and references there.

The LP equations can be derived from the dynamics generated by the (quantum) Fr\"ohlich Hamiltonian for suitable initial states in the strong coupling limit $\alpha\to \infty$ \cite{LMRSS} (see also \cite{FS,FG,Griesemer,LRSS,M} for earlier results on this problem). One of the outstanding open problems concerns the polaron's effective mass \cite{LS,seiringer2020polaron,spohn1987effective}: due to the interaction with the polarization field, the electron effectively becomes heavier and behaves like a particle with a larger mass. This mass increases with the coupling $\alpha$, and is expected to diverge as $\alpha^4$ as $\alpha\to\infty$. A precise asymptotic formula was obtained by Landau and Pekar \cite{ landau1948effective} based on the classical approximation, and hence it is natural to ask to what extent the derivation of the LP equations in \cite{LMRSS} allows to draw conclusions on the effective mass problem.

It is, however, far from obvious how to rigorously obtain the effective mass even on the classical level, i.e., from the LP equations \eqref{eq:LP}. A heuristic derivation, reviewed in Section \ref{sec:TW} below, considers traveling wave solutions of \eqref{eq:LP} for non-zero velocity $v\in \R^3$, and expands the corresponding energy for small $v$. The existence of such solutions remains unclear, however, and we in fact conjecture that no such solutions exist for non-zero $v$. This is related to the fact the energy functional corresponding to \eqref{eq:LP} (given in Eq. \eqref{def:G} below) does not dominate the total momentum, and a computation of the ground state energy as a function of the (conserved) total momentum would simply yield a constant function (corresponding to an infinite effective mass). Due to the vanishing of the sound velocity in the medium, a moving electron can be expected to be slowed down to zero speed by emitting radiation. (See \cite{Fr1,Fr2,Fr3,Fr4,Fr5} for the study of a similar effect in a model of a classical particle coupled to a field.) 

In this paper, we provide a novel definition of the effective mass for the LP equations. We shall argue that all low energy states have a well-defined notion of (initial) velocity, and hence we can minimize the energy functional among states with given velocity. Expanding the resulting energy-velocity relation for small velocity gives a definition of the effective mass, which coincides with the prediction by Landau and Pekar \cite{landau1948effective}. 

\subsection{Structure of the paper}
 In Section \ref{sec:result}, we explain our rigorous approach to derive the energy-velocity relation of the system, allowing for a precise definition and computation of the effective mass. After introducing some notation and recalling fundamental properties of the Pekar energy functional in Section \ref{sec:prel},  we identify in Section \ref{sec:vel} a set of initial data for the LP equations for which it is possible  to define the position, and consequently the velocity, at any time. We then arrive at an energy-velocity relation by defining $E(v)$ in Section \ref{sec:initial} as the minimal energy among all admissible initial states of fixed initial velocity $v$. Finally, in Section \ref{sec:thm} we state our main result, an expansion of $E(v)$ for small velocities $v$, allowing for the computation the effective mass of the system.  

Section \ref{sec:proof-thm} contains the proof of our main result, Theorem \ref{thm:freephase}.  

In Section \ref{sec:approaches} we discuss the formal approach to the effective mass via traveling waves. Moreover, we investigate an alternative definition of the effective mass, through an alternative notion of velocity of low-energy states.

\section{Main Results}
\label{sec:result}

\subsection{Preliminaries}
\label{sec:prel} We start by introducing further notation and recalling some known results. 
The classical energy functional corresponding to the Landau--Pekar equations \eqref{eq:LP} is defined on $H^1(\Rtre)\times L^2(\Rtre)$ as 
\begin{equation}\label{def:G}
\G ( \psi, \varphi ) = \langle \psi, h_\varphi \psi \rangle + \| \varphi \|_2^2 \quad \text{for} \quad \|\psi\|_2=1 .
\end{equation}
Equipped with the symplectic form $\frac 1{2i} \int d\psi \wedge d\bar\psi + \frac{\alpha^2}{2i} \int d\varphi\wedge d\bar\varphi$, it defines a dynamical system leading to the LP equations \eqref{eq:LP}.  Moreover, $\G$ is conserved along solutions of \eqref{eq:LP}.

Let $e_{\rm P}$ denote the Pekar ground state energy 
\begin{equation}\label{def:ep}
 e_{\rm P} := \min \G(\psi,\varphi) \,.
\end{equation}
(For an estimation of its numerical value, see \cite{miyake}.) 
It was proved in \cite{lieb1977existence} that the minimum in \ref{def:ep} 
is attained for the Pekar minimizers $( \psi_{\rm P}, \varphi_{\rm P})$, which are radial smooth functions in $C^{\infty}(\Rtre)$ satisfying $\psi_{\rm P}>0$, $\varphi_{\rm P}=-\sigma_{\psi_{\rm P}}$ and $\psi_{\rm P}=\psi_{\varphi_{\rm P}}$, where  $\psi_{\varphi}$ denotes the ground state of $h_{\varphi}$ whenever it exists. Moreover, this minimizer is unique up to the symmetries of the problem, i.e., translation-invariance and multiplication of $\psi$ by a phase. 
We shall denote 
\begin{align}
\label{eq:def-HP}
H_{\rm P} = h_{\varphi_{\rm P}} - \mu_{\rm P}  \quad \text{with} \quad \mu_{\rm P}=\infspec h_{\varphi_{\rm P}} .
\end{align}  

Associated to $\G$, there are the two functionals 
\begin{align}
\label{eq:E}
\E(\psi):=\inf_{\varphi\in L^2(\Rtre)} \G(\psi,\varphi) = \int_{\R^3} |\nabla \psi(x)|^2 dx - \frac{ 1}{4\pi} \int_{\R^6} \frac{ |\psi(x)|^2 |\psi(y)|^2}{|x-y|} dx\, dy
\end{align}
and
\begin{align}
\label{eq:F}
 \F(\varphi):=\inf_{\psi\in H^1(\Rtre) \atop \|\psi\|_2=1} \G(\psi,\varphi) = \infspec h_\varphi + \|\varphi\|_2^2
\end{align}
%
and clearly $e_{\rm P}=\min \G(\psi,\varphi)=\min\E(\psi)=\min \F(\varphi)$. 
We also define the manifolds of minimizers 
\begin{align}
\label{eq:PekEn}
\minG:=\{(\psi,\varphi) \,|\, \G(\psi,\varphi)&=e_{\rm P}\}, \ \minE:=\left\{\psi \,|\, \E(\psi)=e_{\rm P}\right\}, \ \minF:=\left\{\varphi \,|\, \F(\varphi)=e_{\rm P}\right\}.
\end{align}
The results  in \cite{lieb1977existence} imply that we can write these in terms of the Pekar minimizers $(\psi_{\rm P}, \varphi_{\rm P})$ as 
\begin{align}
\minG=\{(e^{i\theta}& \psi_{\rm P}^y,\varphi_{\rm P}^y) \,|\, \theta \in [0,2\pi), \, y \in \Rtre \},\nonumber\\
\minE=\{&e^{i\theta} \psi_{\rm P}^y \,|\, \theta \in [0,2\pi), \, y \in \Rtre\},\nonumber\\
&\minF=\{\varphi_{\rm P}^y \,|\, y \in \Rtre\}
\end{align}
where  $f^y:=f(\,\cdot\,-y)$ for any function $f$.  Furthermore, it can be deduced from the results in  \cite{lenzmann2009uniqueness}  that the energy functionals $\F$ and $\E$ are both coercive (see \cite[Lemmas~2.6 and 2.7]{feliciangeli2021persistence}), i.e., there exists $C>0$ such that
\begin{align}
\label{eq:coercivity}
\F(\varphi)\geq e_{\rm P}+C\dist^2_{L^2}(\varphi,\minF), \quad \E(\psi)\geq e_{\rm P}+ C \dist^2_{H^1}(\psi,\minE).
\end{align}

The following Lemma on properties of the projection onto  the manifold $\minF$ will be important for our analysis below. Its proof will be given in Appendix \ref{app:projection}.

\begin{lem}
	\label{lem:minEFproj}
	There exists $\delta>0$ such that the $L^2$-projection onto $\minF$, is well-defined (i.e., unique) on 
	\begin{align}
	(\minF)_{\delta}&:=\{\varphi\in L^2(\Rtre)\;|\; \dist_{L^2}(\varphi, \minF)\leq \delta\} \; .
	\end{align}
	For any $\varphi \in (\minF)_{\delta}$,  we define $z_\varphi \in \R^3$ via 
	\begin{align}
	P_{L^2}^{\minF}(\varphi)&=\varphi_{\rm P}^{z_{\varphi}} \; .
	\end{align}
	Then $z_{\varphi}$ is a differentiable function from $(\minF)_{\delta}$ to $\Rtre$ and its partial derivative in the direction $\eta \in L^2(\R^3)$  is given by 
	\begin{align}
	 \partial_t z_{\varphi + t \eta} \restriction_{t=0} 
	 = A_\varphi^{-1} \langle \Re \eta \vert  \nabla \varphi_P^{z_\varphi} \rangle,
	\end{align} 
	where $A$ is the invertible matrix defined  for any $\varphi \in(\minF)_{\delta}$  by $A_{i,j} :=  - \Re \langle \varphi \vert \partial_i \partial_j \varphi_P^{z_\varphi} \rangle$. 
\end{lem}

\begin{rem} \label{rem:E} Likewise, it can be shown that the $H^1$- (resp. $L^2$-) projection onto $\minE$ have similar properties: There exists $\delta >0 $ such that the $H^1$- (resp. the $L^2$-) projection onto $\minE$ 
	\begin{align}
	\label{eq:proj-E}
	P_{H^1}^{\minE}(\psi)=e^{i\theta_{\psi}}\psi_{\rm P}^{y_{\psi}}, \quad\Big(\text{resp.}  \ P_{L^2}^{\minE}(\psi)&=e^{i\theta'_{\psi}}\psi_{\rm P}^{y'_{\psi}} \Big) 
	\end{align}
	is well-defined on the set $(\minE)^{H^1}_{\delta}:=\{\psi\in L^2(\Rtre)\;|\; \dist_{H^1}(\psi, \minE)\leq \delta\}$ (resp.  $(\minE)^{L^2}_{\delta}:=\{\psi\in L^2(\Rtre)\;|\; \dist_{L^2}(\psi, \minE)\leq \delta\}$) and the functions $y_\psi, \theta_\psi$ (resp.  $y'_\psi, \theta'_\psi$) defined through  \eqref{eq:proj-E} 
are differentiable functions from  $(\minE)^{H^1}_{\delta}$ (resp.  $(\minE)^{L^2}_{\delta}$)  to $\Rtre$ and $\mathbb{R} / (2\pi \mathbb{Z})$. 
\end{rem}

\subsection{Position and velocity of solutions} \label{sec:vel}

In this section, we give a meaning to the notion of position, and therefore velocity, for solutions of the Landau--Pekar equations (at least for a class of initial data). There is a natural way of defining, given $\psi_t$, the position of the electron at time $t$, which is simply given by 
\begin{align}
\label{eq:defXel}
\Pel(t):=\expval{x}{\psi_t}.
\end{align}
This yields, by straightforward computations using \eqref{eq:LP}, that
\begin{align}\label{vel}
\Vel(t):=\frac{d}{dt}\Pel(t)=2\expval{-i\nabla}{\psi_t}.
\end{align} 
Note that \eqref{vel} is always well-defined for $\psi\in H^1(\R^3)$, even although \eqref{eq:defXel} not necessarily is. 

For the phonon field, the situation is more complicated as $\varphi$ cannot be interpreted as a probability distribution over positions. This calls for a different approach. By \eqref{eq:coercivity}, Lemma \ref{lem:minEFproj} and the conservation of $\G$ along solutions of \eqref{eq:LP}, we know that there exists $\delta^*$ such that for any initial condition $(\psi_0,\varphi_0)$ such that 
\begin{align}
\label{eq:IDuniqueproj}
\G(\psi_0,\varphi_0)\leq e_{\rm P}+\delta^*,
\end{align}
$\varphi_t$ 
admits a unique $L^2$-projection $\varphi_{\rm P}^{z(t)}$ onto $\minF$ for all times.
We use this to define 
\begin{align}
\label{eq:defXph}
\Pph(t):=z(t), \quad \Vph:=\frac{d}{dt} \Pph(t)=\dot{z}(t).
\end{align}
Note that $\Pph(t)$ is indeed differentiable by Lemma \ref{lem:minEFproj} and the differentiability of the LP dynamics. At this point, for any initial data satisfying \eqref{eq:IDuniqueproj}, we have a well-defined notion of position and velocity for all times, admittedly in a much less explicit form for the phonon field.

\subsection{Initial conditions of velocity $v$}
\label{sec:initial}
 
For any $v\in \Rtre$ (or at least for $|v|$ sufficiently small), we are now interested in considering all initial conditions $(\psi_0,\varphi_0)$ whose solutions have instantaneous velocity $v$ at $t=0$ (both in the electron and in the phonon coordinate) and to then minimize the functional $\G$ over such states. This will give us an explicit relation between the energy and the velocity of the system, allowing us to define the effective mass of the polaron in the classical setting defined by the Landau--Pekar equations.

Note that by radial symmetry of the problem only the absolute value of the velocity, and not its direction, affects our analysis. Hence, for $v\in \R$, we consider  initial conditions $(\psi_0,\varphi_0)$ such that
\begin{itemize}
	\item[(i)] $(\psi_0, \varphi_0) \in H^1 (\R^3) \times L^2 ( \R^3) $ with $\| \psi_0\|_2=1$ and such that \eqref{eq:IDuniqueproj} is satisfied,
	\item[(ii)] $\Vel(0)=\Vph(0)=v (1,0,0)$. 
\end{itemize}
The set of admissible initial conditions of velocity $v\in \R$ can hence be compactly written as
\begin{align}
\label{eq:defI}
I_v:= \lbrace (\psi_0, \varphi_0) \, \vert \,  \text{(i), (ii) are satisfied} \rbrace.
\end{align}
We will show below that it is non-empty for small enough $v$. 

\subsection{Expansion of the energy} \label{sec:thm} In order to compute the effective mass of the polaron, we now minimize the energy $\G$ over the set $I_v$.
To this end, we define the energy
\begin{align}\label{def:Ev}
E(v) :=\inf_{(\psi_0,\varphi_0)\in I_v} \G(\psi_0,\varphi_0).
\end{align}
The following theorem gives an expansion of $E(v)$ for sufficiently small velocities $v$.  Its proof will be given in Section \ref{sec:proof-thm}.

\begin{theo}
	\label{thm:freephase}
	As $v \rightarrow 0$  we have 
	\begin{equation}
	\label{eq:thm}
	E(v)=e_{\rm P}+v^2 \left(  \frac{1}{4}+ \frac{\alpha^4}{3} \| \nabla  \varphi_{\rm P} \|_2^2 \right) + O(v^3).	
	\end{equation}
\end{theo}

Since the kinetic energy of a particle of mass $m$ and velocity $v$ equals $m v^2/2$, \eqref{eq:thm} identifies the effective mass of the system as
\begin{align}
\label{eq:OurEffMass}
m_{\text{eff}}= \lim_{v \rightarrow 0} \frac{E(v) - e_{\rm P}}{v^2/2} = \frac 1 2 +\frac {2\alpha^4} 3 \| \nabla \varphi_{\rm P} \|_2^2  .
\end{align}
The first term $1/2$ is simply the bare mass of the electron in our units, while the second term $\frac {2\alpha^4} 3 \| \nabla \varphi_{\rm P} \|_2^2$ corresponds to the additional mass acquired through the interaction with the  phonon field. It agrees with the prediction in \cite{landau1948effective}, and is conjectured to coincide with the effective mass in the Fr\"ohlich model in the limit $\alpha \to \infty$.  Note that since $(-\Delta)^{1/2} \varphi_{\rm P} = -|\psi^{\rm P}|^2$, $\|\nabla \varphi_{\rm P}\|_2 = \|\psi^{\rm P}\|_4^2$, which can be evaluated numerically \cite{miyake}. 

\begin{rem}[Traveling waves]
\label{rem:TW}
The heuristic computations contained in the physics literature concerning $m_{\text{eff}}$ \cite{alexandrov2010advances,landau1948effective}
all rely, in one way or another, on the existence of traveling wave solutions of the LP equations of velocity $v$ (at least for sufficiently small velocity), i.e. solutions with initial data $(\psi_v,\varphi_v)$ such that 
\begin{align}
(\psi_t(x),\varphi_t(x))=(e^{-ie_{v}t}\psi_v(x-vt),\varphi_v(x-vt))
\end{align}
for suitable $e_v \in \R$. 
Such solutions would allow to define the energy of the system at velocity $v$ as $E^{\text{TW}}(v) = \G( \psi_v, \varphi_v)$, and a perturbative calculation  (discussed in Section \ref{sec:TW} below) yields indeed
\begin{align}
\label{eq:mass-predTW}
\lim_{v \rightarrow 0} \frac{E^{\text{TW}}(v) - e_{\rm P}}{v^2/2} = \frac 1 2 +\frac {2\alpha^4} 3 \| \nabla \varphi_{\rm P} \|^2, 
\end{align} 
in agreement with \eqref{eq:OurEffMass}. 
 Unfortunately, this approach turns out to be only formal, and we conjecture traveling wave solutions to not exist for any $\alpha>0$, $v>0$, as explained in more detail in Section~\ref{sec:TW}. 
\end{rem}

\begin{rem} \label{rem:AlternatePos}
In Section \ref{sec:vel}, we used the standard approach from quantum mechanics to define the electron's position \eqref{eq:defXel} and velocity  \eqref{vel}. We could, instead, use also for the electron a similar approach to the one we use for the phonon field (i.e. \eqref{eq:defXph}) through the projection onto the manifold of minimizers $\minE$.  A natural question is whether one  obtains the same effective mass using this different notion of position. In Section \ref{sec:diff-vel}, we show that, in fact, this alternate definition yields a different effective mass equal to 
\begin{align}
\widetilde{m}_{\text{eff}} = \frac{2 \| \nabla \psi_{\rm P} \|_2^4}{3 \|\nabla \varphi_{\rm P} \|_2^2 } +\frac{2\alpha^4}{3} \| \nabla \varphi_{\rm P} \|_2^2  .
\end{align}
This coincides with \eqref{eq:OurEffMass} and \eqref{eq:mass-predTW} for large $\alpha$ (hence still confirming the prediction in \cite{landau1948effective}), but differs  in the $O(1)$ term. In fact, as we discuss in Section \ref{sec:diff-vel}, one has $\widetilde{m}_{\text{eff}} < m_{\text{eff}}$. 
\end{rem}

\section{Proof of Theorem~\ref{thm:freephase}}\label{sec:proof-thm}

%

 Let us denote $\delta_1=\psi_0-\psi_{\rm P}$ and $\delta_2=\varphi_0-\varphi_{\rm P}$. Expanding $\G$ in \eqref{def:G} and using that $\varphi_{\rm P}=-\sigma_{\psi_{\rm P}}$ we find
\begin{align}
\G(\psi_0,\varphi_0) & =\G(\psi_{\rm P}+\delta_1,\varphi_{\rm P}+\delta_2) \nonumber\\
& =e_{\rm P}+2\bra{\psi_{\rm P}}h_{\varphi_{\rm P}}\ket{\Re \delta_1}\nonumber\\
&\quad +\expval{h_{\varphi_{\rm P}}}{\delta_1}+2\bra{\Re \delta_1}V_{\delta_2}\ket{\psi_{\rm P}}+\|\delta_2\|_2^2 +\expval{V_{\delta_2}}{\delta_1}. \label{eq:expansion-G}
\end{align}
Since $\psi_0$ is normalized, we have
\begin{align}\label{pdd}
1=\|\psi_0\|_2^2=\|\psi_{\rm P}+\delta_1\|_2^2=1+\|\delta_1\|_2^2+2\bra{\psi_{\rm P}}\ket{\Re \delta_1}\iff 2\bra{\psi_{\rm P}}\ket{\Re \delta_1}=-\|\delta_1\|_2^2.
\end{align}
Hence
\begin{align}
2\bra{\psi_{\rm P}}h_{\varphi_{\rm P}}\ket{\Re \delta_1}=2\mu_{\rm P}\bra{\psi_{\rm P}}\ket{\Re \delta_1}=-\mu_{\rm P}\|\delta_1\|_2^2,
\end{align}
and  using $\| V_{\delta_2} \delta_1 \|_2 \leq C \| \delta_2\|_2 \; \|\delta_1 \|_{H^1}$ (see, e.g., \cite[Lemma III.2]{LRSS}) we arrive at
\begin{align}
\G(\psi_0,\varphi_0)=e_{\rm P}+\expval{H_{\rm P}}{\delta_1}+2\bra{\Re \delta_1}V_{\delta_2}\ket{\psi_{\rm P}}+\|\delta_2\|_2^2+O(\|\delta_2\|_2 \|\delta_1\|_{H^1}^2).
\end{align}
By completing the square, we have
\begin{align}
&\|\Re\delta_2\|_2^2+2\bra{\Re \delta_1}V_{\delta_2}\ket{\psi_{\rm P}}\nonumber\\
&=\|\Re \delta_2 +2 
(-\Delta)^{-1/2} (\psi_{\rm P}\Re \delta_1)\|_2^2-4 
\expval{\psi_{\rm P}(-\Delta)^{-1} \psi_{\rm P}}{\Re \delta_1}
\end{align}
and therefore
\begin{align}\label{36}
\G(\psi_0,\varphi_0) &= e_{\rm P}+\expval{H_{\rm P}}{\Im\psi_0}+\|\Im \varphi_0\|_2^2 \nonumber \\ & \quad +\|\Re \delta_2 +2 
(-\Delta)^{-1/2} (\psi_{\rm P}\Re \delta_1)\|_2^2\nonumber\\
&\quad +\expval{H_{\rm P}-4X_{\rm P}}{\Re \delta_1}+O(\|\delta_2\|_2 \|\delta_1\|^2_{H^1}) ,
\end{align} 
where $X_{\rm P}$ is the operator with integral kernel $X_{\rm P}(x,y):= 
\psi_{\rm P}(x)(-\Delta)^{-1}(x,y) \psi_{\rm P}(y)$. Since $X_{\rm P}$ is bounded, and $\| P_{\psi_{\rm P}} \Re\delta_1\| = \|\delta_1\|_2^2/2$ by \eqref{pdd} (with $P_{\psi_{\rm P}} = |\psi_{\rm P}\rangle\langle \psi_{\rm P}|$ the rank one projection onto $\psi_{\rm P}$), 
we also have 
%
\begin{align}
\G(\psi_0,\varphi_0)& =e_{\rm P}+\expval{H_{\rm P}}{\Im\psi_0}+\|\Im\varphi_0\|_2^2 \nonumber \\ & \quad +\|\Re \delta_2 +2
(-\Delta)^{-1/2} (\psi_{\rm P}\Re \delta_1)\|_2^2\nonumber\\
 &\quad+\expval{Q(H_{\rm P}-4X)Q}{\Re \delta_1}+O(\|\delta_2\|_2\|\delta_1\|^2_{H^1})+O(\|\delta_1\|_{L^2}^3) , \label{eq:expressionG}
\end{align}
where $Q=\unit-P_{\psi_{\rm P}}$.

\textbf{Upper Bound:} For sufficiently small $v$,  we use as a trial state
\begin{align}
\label{eq:trialstate}
(\bar\psi_0,\bar\varphi_0)=\left( f_v \psi_{\rm P}+i g_v H_{\rm P}^{-1}\partial_1 \psi_{\rm P} , \; \varphi_{\rm P}+iv\alpha^2 \partial_1 \varphi_{\rm P}\right)
\end{align}
with $f_v, g_v >0$ given by 
\begin{align}
\label{eq:defFG}
f_v^2 := \frac{2 v^2 \|H_{\rm P}^{-1}\partial_1 \psi_{\rm P}\|_2^2 }{1 - \sqrt{1-4v^2 \|H_{\rm P}^{-1}\partial_1 \psi_{\rm P}\|_2^2 } }, \;  g_v^2 := \frac{ 1 - \sqrt{1-4v^2 \|H_{\rm P}^{-1}\partial_1 \psi_{\rm P}\|_2^2 } }{2 \|H_{\rm P}^{-1}\partial_1 \psi_{\rm P}\|_2^2 }  .
\end{align}
Note that $\partial_1 \psi_{\rm P}$ is orthogonal to $\psi_{\rm P}$, hence $H_{\rm P}^{-1} \partial_1 \psi_{\rm P}$ is well-defined. 
We begin by showing that \eqref{eq:trialstate} is an element of the set of admissible initial data $I_v$ in \eqref{eq:defI}. To prove that $(\bar\psi_0,\bar\varphi_0)$ satisfies (i), we only need to check that $\bar\psi_0$ is normalized (which follows easily from \eqref{eq:defFG}) as the  condition \eqref{eq:IDuniqueproj}  will follow a posteriori from the energy bound we shall derive. We now proceed to show that $(\bar\psi_0,\bar\varphi_0)$ satisfies (ii). For the electron, using that  $H_{\rm P}^{-1} \partial_j \psi_{\rm P}= -x_j \psi_{\rm P} /2  $ (which can be checked by applying $H_P$ and using that $[H_P,x_1]=-2 \partial_1$) and  consequently that 
\begin{align}
\label{eq:sandwich}
\bra{\partial_i \psi_{\rm P}}H_{\rm P}^{-1}\ket{\partial_j \psi_{\rm P}}=\delta_{ij}/ 4 
\end{align} 
since $\psi_{\rm P}$ is radial, 
we can conclude that
\begin{align}
\label{eq:trial-v}
-2\expval{i \partial_j }{\bar \psi_0}=4 f_v g_v  \langle H_{\rm P}^{-1} \partial_1 \psi_{\rm P} \vert \partial_j \psi_{\rm P} \rangle  =v \delta_{j1}  ,
\end{align}
i.e., that $\Vel(0)=v(1,0,0) $, as required.

For the phonons, we first note that $\Pph(0)=0$, since $\Re \bar{\varphi}_0 = \varphi_{\rm P}$. Next,  we derive a relation for the velocity of the phonons $\Vph(t) = \dot{z}(t)$ in terms of their position $\Pph (t) =z(t)$ for general time $t$.  Since 
\begin{align}
\label{eq:defZ}
\min_z\|\varphi_t -\varphi_{\rm P}^z\|_2^2=\|\varphi_t-\varphi_{\rm P}^{z(t)}\|_2^2,
\end{align}
the position $z(t)$ necessarily  has to satisfy 
\begin{align}
\label{eq:orthPhon}
\Re \bra{\varphi_t}\ket{(u\cdot \nabla) \varphi_{\rm P}^{z(t)}}=0 \quad \text{for all} \quad u \in \mathbb{S}^2 \iff \Re \varphi_t \perp \spn\{\nabla \varphi_{\rm P}^{z(t)}\}.
\end{align}
Differentiating \eqref{eq:orthPhon} w.r.t. time, using \eqref{eq:LP} and evaluating the resulting expression at $t=0$,  we arrive at 
\begin{align}
\label{eq:defVph2}
0=&\Re\bra{-i\alpha^{-2} (u\cdot \nabla)(\bar\varphi_0+\sigma_{\bar\psi_0})}\ket{\varphi_{\rm P}}-\Re\bra{(\dot{z}(0)\cdot\nabla)\bar\varphi_0}\ket{(u\cdot \nabla)\varphi_{\rm P}}\nonumber\\
=&\bra{-\alpha^{-2} \Im\bar\varphi_0}\ket{(u\cdot \nabla)\varphi_{\rm P}}-\bra{(\dot{z}(0)\cdot\nabla)\Re \bar\varphi_0}\ket{(u\cdot \nabla)\varphi_{\rm P}}\nonumber\\
=&-\bra{v\partial_1 \varphi_{\rm P}}\ket{(u\cdot \nabla)\varphi_{\rm P}}-\bra{(\dot{z}(0)\cdot\nabla)\varphi_{\rm P}}\ket{(u\cdot \nabla)\varphi_{\rm P}},
\end{align}
which the velocity $\dot{z} (0)$ has to satisfy for all $u\in \mathbb{S}^2$, given its position $\Pph(0)=z (0)=0$.  By invertibility of the coefficient matrix,  \eqref{eq:defVph2} has the unique solution $\dot{z}(0)=v(1,0,0)$, and we indeed conclude that $\Vph (0) = v (1,0,0)$.

We now evaluate $\G(\bar\psi_0,\bar\varphi_0)$. Since $f_v = 1 + O(v^2), g_v = v+ O(v^3)$, using \eqref{eq:expressionG} and \eqref{eq:sandwich} we find
\begin{align}
E(v)\leq \G(\bar\psi_0,\bar\varphi_0)=e_{\rm P}+v^2 \left(  \frac 1 4+\alpha^4 \| \partial_1 \varphi_{\rm P} \|_2^2 \right) + O(v^3)
\end{align}
verifying on the one hand \eqref{eq:IDuniqueproj} for sufficiently small $v$, and on the other hand the r.h.s. of \eqref{eq:thm} as an upper bound on $E(v)$ (using that $\varphi_{\rm P}$ is radial).

\textbf{Lower Bound:} We first observe that to derive a lower bound on $E(v)$ we can w.l.o.g. restrict to initial conditions $(\psi_0,\varphi_0)$  satisfying  additionally
\begin{align}
\label{eq:centeringI}
P_{L^2}^{\minE}(\psi_0)>0, \\
\label{eq:centeringII} \Pph(0)=0.
\end{align} 
This simply follows from the invariance of $\G$ under translations (of both $\psi$ and $\varphi$) and under changes of phase of $\psi$.  
Moreover, by the upper bound derived in the first step of this proof and the coercivity of $\E$ and $\F$ in \eqref{eq:coercivity}, we conclude that it is sufficient to minimize over elements of $I_v$ such that $\dist_{H^1}(\psi_0, \minE)=O(v)=\dist_{L^2}(\varphi_0,\minF)$  for small $v$.  Since the $L^2$-projection of $\varphi_0$ is $\varphi_{\rm P}$ by \eqref{eq:centeringII}, it immediately follows that $\|\delta_2\|_2=O(v)$. We now proceed to show that necessarily also $\|\delta_1\|_{H^1}=O(v)$. 
Let $y',y\in \Rtre$ and $\theta \in[0,2\pi)$ be such that 
\begin{align}
\label{eq:defTheta}
P_{L^2}^{\minE}(\psi_0)=\psi_{\rm P}^{y'}, \quad P_{H^1}^{\minE}(\psi_0)=e^{i\theta}\psi_{\rm P}^y,
\end{align}
where we recall that the $L^2$-projection is assumed to be positive by \eqref{eq:centeringI}. Combining the upper bound derived in the first step with  \cite[Eq.~(53)]{feliciangeli2021persistence}, we get
\begin{align}
\| \varphi_0 - \varphi_{\rm P}^y \|_2^2 \leq C \left( \G (\psi_0, \varphi_0)  - e_{\rm P} \right)\leq Cv^2.
\end{align}
There exist $\delta, C_1, C_2>0$ such that  
\begin{align}
\| \varphi_{\rm P} - \varphi_{\rm P}^y \|_2\geq\begin{cases}
C_1\vert y \vert \| \nabla \varphi_{\rm P} \|_2, & |y|\leq \delta\\
C_2 & |y|>\delta
\end{cases},
\end{align}
and this allows to conclude that $|y| = O(v)$. In other words,  assuming  centering w.r.t. to translations in the phonon coordinate (i.e. \eqref{eq:centeringII}) forces, at low energies, also the centering w.r.t. translations in the electron coordinate, at least approximately. At this point, it is also easy to verify that $\theta=O(v)$ (and, as an aside, that $|y'|=O(v)$), since, by the upper bound derived in the first step and the coercivity of $\E$, we have
\begin{align}
\|\psi_{\rm P}^{y'}-e^{i\theta} \psi_{\rm P}^y\|_2\leq \|\psi_{\rm P}^{y'}-\psi_0\|_{2}+\|e^{i\theta} \psi_{\rm P}^y-\psi_0\|_{2}=O(v).
\end{align}
In particular, we conclude that
\begin{align}\label{aa}
\| \delta_1 \|_{H^1} \leq \| e^{i\theta}\psi_{\rm P}^{y} - \psi_0 \|_{H^1}+ \| \psi_{\rm P} - e^{i\theta}\psi_{\rm P}^{y} \|_{H^1} = O(v).
\end{align} 

Using again \eqref{eq:expressionG} and that $Q(H_{\rm P}-4X_{\rm P})Q\geq 0$, we conclude that for any $(\psi_0,\varphi_0)\in I_v$ satisfying 
\eqref{eq:centeringI} and \eqref{eq:centeringII}, as well as $\G(\psi_0,\varphi_0) \leq e_{\rm P} + O(v^2)$, we have
\begin{align}
\label{eq:firstLB}
\G(\psi_0,\varphi_0)\geq e_{\rm P}+\expval{H_{\rm P}}{\Im\psi_0}+\|\Im\varphi_0\|_2^2+O(v^3).
\end{align}
By arguing as in \eqref{eq:defVph2}, we see that the conditions $\Pph(0)=0$ and $\Vph(0)=v$ imply that 
%
\begin{align}
P_{\nabla \varphi_{\rm P}} (\Im \varphi_0 + v\alpha^2 \partial_1 \Re\varphi_0)=0 ,
\end{align}
where $P_{\nabla \varphi_{\rm P}}$ denotes the projection onto the span of $\partial_j \varphi_{\rm P}$, $1\leq j\leq 3$. 
Since $P_{\nabla \varphi_{\rm P}} \partial_1$ is a bounded operator, and $\|\delta_2\|_2=O(v)$, we find 
\begin{equation}\label{325}
\|\Im\varphi_0\|_2^2 \geq \| P_{\nabla \varphi_{\rm P}} \Im\varphi_0\|_2^2 = v^2 \alpha^4 \| \partial_1 \varphi_{\rm P} + P_{\nabla \varphi_{\rm P}} \partial_1 \Re \delta_2 \|_2^2 \geq v^2 \alpha^4 \| \partial_1 \varphi_{\rm P}\|_2^2 - O(v^3) .
\end{equation}
%
We are left with giving a lower bound on $\expval{H_{\rm P}}{\Im\psi_0}$,  
under the condition that   
\begin{align}
2\expval{-i\nabla}{\psi_0}=4\bra{\Im \psi_0}\ket{\nabla\Re\psi_0}=v(1,0,0).
\end{align}
We already argued in \eqref{aa} that  $\|\psi_0-\psi_{\rm P}\|_{H^1}=O(v)$, and therefore
\begin{align}
\label{eq:constranit-Im}
4\bra{\Im \psi_0}\ket{\nabla\psi_{\rm P}}=v(1,0,0)+O(v^2).
\end{align}
Completing the square, we find 
\begin{align}
\label{eq:completing-square}
\expval{H_{\rm P}}{\Im\psi_0} & = \expval{H_{\rm P}^{-1}}{H_{\rm P}\Im\psi_0 - v \partial_1 \psi_{\rm P}} \notag \\
& \quad + 2v \bra{\Im\psi_0}\ket{\partial_1  \psi_{\rm P}} - v^2 \expval{H_{\rm P}^{-1}}{ \partial_1  \psi_{\rm P}} \notag \\
&  \geq  2v \bra{\Im\psi_0}\ket{\partial_1  \psi_{\rm P}} - v^2 \expval{H_{\rm P}^{-1}}{ \partial_1 \psi_{\rm P}} .
\end{align}
From the constraint \eqref{eq:constranit-Im} and \eqref{eq:sandwich}, it thus follows that
\begin{align}\label{329}
\expval{H_{\rm P}}{\Im\psi_0} &\geq  \frac{v^2}{4}  + O(v^3).
\end{align}
By combining \eqref{eq:firstLB}, \eqref{325} and \eqref{329}, we arrive at the final lower bound 
\begin{align}
E(v) \geq e_{\rm P}+ v^2\left(\frac 1 4 + \alpha^4 \|\partial_1\varphi_{\rm P}\|_2^2\right)+O(v^3) .
\end{align}
Again,  since $\varphi_{\rm P}$ is radial, this is of the desired form, and hence the proof is complete.
%
%
\hfill\qed

\section{Further Considerations}
\label{sec:approaches}

In this section, we carry out the details related to Remarks \ref{rem:TW} and \ref{rem:AlternatePos}. 

\subsection{Effective mass through traveling wave solutions}
\label{sec:TW}

A possible way of formalizing the derivation of the effective mass in \cite{alexandrov2010advances,landau1948effective}  relies on traveling wave solutions of the Landau--Pekar equations.  
A traveling wave of velocity $v\in \Rtre$ 
is a solution $(\psi_t, \varphi_t)$ of \eqref{eq:LP} 
of the form 
	\begin{align}
	\label{eq:TWdef}
	(\psi_t, \varphi_t) = (e^{-i e_v t} \psi_v^{\text{TW}} ( \cdot - vt ), \varphi_v^{\text{TW}} ( \cdot - vt ) ) 
 \end{align} 
 for all $t\in \R$, with $e_v\in \R$ defining a suitable phase factor.  As before,  by rotation invariance we can restrict our attention to velocities of the form $v (1,0,0)$ with $v\in \R$ in the following.
 
Note that in the case $\alpha =0$, where $\varphi_t = - \sigma_{\psi_t}$ for all $t\in\R$, the LP equations simplify to a non-linear Schr\"odinger equation (also known as Choquard equation). In this case,  a traveling wave is given by $\psi_v^{\text{TW}} (x)= e^{i x_1 v /2 } \psi_{\rm P} ( x)$ with $e_v= \mu_{\rm P} + \frac{v^2}{4} $, and its energy can be computed to be 
\begin{align}
\G\left( \psi_v^{\text{TW}}, -\sigma_{\psi_v^{\text{TW}}} \right) = e_{\rm P} + \frac{v^2}{4}, \label{eq:Ev-expansion-alpha-0}
\end{align}
yielding an effective mass $m=1/2$ at $\alpha=0$. For the case $\alpha > 0$, on the other hand, we conjecture that there are no traveling wave solutions of the form \eqref{eq:TWdef}.  
 
 \begin{Conj}
 For $\alpha>0$, there are no solutions to the LP equations \eqref{eq:LP} of the form \eqref{eq:TWdef} with $v\neq 0$.
 \end{Conj}
 
The motivation for this conjecture comes from the vanishing of the sound velocity in the medium. An initially moving electron can be expected to be slowed down to zero speed by emitting radiation. Establishing this effect rigorously for the LP equations remains an open problem, however.
 
 If one {\em assumes} the existence of traveling wave solutions, at least for small $v$, one can predict an effective mass that agrees with our formula \eqref{eq:OurEffMass}, as we shall now demonstrate.
From the LP equations \eqref{eq:LP} one easily sees that  a traveling wave solution needs to satisfy 
\begin{equation}\label{eq:LPTW}
\begin{aligned}
-i v \partial_1  \psi_v^{\text{TW}}  & =\left(  h_{ \varphi_v^{\text{TW}}} + e_v \right)   \psi_v^{\text{TW}}  \\
-i \alpha^2 v \partial_1 \varphi_v^{\text{TW}}  & = \varphi_v^{\text{TW}} + \sigma_{\psi_v^{\text{TW}}} .
\end{aligned}
\end{equation}
We shall denote by $E^{\text{TW}}(v)$ the  energy of the traveling wave as a function of the velocity $v\in \R$, i.e. 
\begin{align}
\label{eq:Evdef}
E^{\text{TW}}(v):=\G(\psi_v^{\text{TW}},\varphi_v^{\text{TW}}) .
\end{align}
In the following, we assume that $e_v=\mu_{\rm P}+O(v^2)$ and that the traveling wave is of the form 
\begin{align}
\label{eq:ansatz_TW}
(\psi_v^{\text{TW}}, \varphi_v^{\text{TW}}) =   \left( \frac{\psi_{\rm P} + v \xi_v}{\| \psi_{\rm P} + v \xi_v \|_2}, \varphi_{\rm P} + v \eta_v \right) ,
\end{align}  
with both $\xi_v$ and $\eta_v$ bounded in $v$ and converging to some $(\xi,\eta)$ as $v\to 0$.  In other words, we assume that the traveling waves have a suitable differentiability in $v$, at least for small $v$, and converge to the standing wave solution $(e^{-i \mu_{\rm P} t}\psi_{\rm P}, \varphi_{\rm P})$ for $v=0$. W.l.o.g. we may also assume that $\xi_v$ is orthogonal to $\psi_{\rm P}$. 

We can then use that 
\begin{align}
\label{eq:exp_norm}
\frac{1}{\| \psi_{\rm P} + v \xi_v \|_2^2} &= 1- v^2  \frac{ \| \xi_v \|_2^2}{\| \psi_{\rm P} + v \xi_v\|_2^2} 
= 1 - v^2 \| \xi \|_2^2 + o (v^2) 
\end{align}
in order to linearize the traveling wave equations \eqref{eq:TWdef}, obtaining that $(\xi, \eta)$ solves
\begin{align}
\label{eq:xieta}
\begin{pmatrix}
i \partial_1 \psi_{\rm P} \\ i \alpha^2 \partial_1 \varphi_{\rm P} 
\end{pmatrix}= 
\begin{pmatrix}
H_{\rm P} & 2
\psi_{\rm P}(-\Delta)^{-1/2} \Re\\ 2
(-\Delta)^{-1/2} \psi_{\rm P}   \Re & 1
\end{pmatrix}
\begin{pmatrix}
\xi \\ \eta
\end{pmatrix},
\end{align}
where $H_{\rm P}=h_{\varphi_{\rm P}}-\mu_{\rm P}$, as defined in \eqref{eq:def-HP}. Splitting into real and imaginary parts, we equivalently find 
\begin{align}
&H_{\rm P} \Im \xi =  \partial_1 \psi_{\rm P} \label{eq:ne_Imxi}\\
&\Im \eta =\alpha^2 \partial_1 \varphi_{\rm P} \label{eq:ne_Imeta} \\
&H_{\rm P} \Re \xi + 2 
 \psi_{\rm P} \left( - \Delta \right)^{-1/2} \Re \eta  = 0 \label{eq:ne_Rexi}\\
&2 
\left( -\Delta \right)^{-1/2} \psi_{\rm P} \Re \xi + \Re \eta  =0. \label{eq:ne_Reeta}
\end{align}
Combining \eqref{eq:ne_Rexi} and \eqref{eq:ne_Reeta} gives $(H_{\rm P} - 4 X_{\rm P}) \Re \xi =0$, with $X_{\rm P}$ defined after \eqref{36}. It was shown in  \cite{lenzmann2009uniqueness} that the kernel of $H_{\rm P} - 4 X_{\rm P}$ is spanned by $\nabla \psi_{\rm P}$, hence $\Re \xi\in \spn \{\nabla \psi_{\rm P}\}$. Eq, 
%
\eqref{eq:ne_Reeta} then implies that $\Re \eta\in \spn \{ \nabla \varphi_{\rm P}\}$. 

Using these equations and \eqref{eq:exp_norm} in  the expansion \eqref{eq:expressionG}, it is  straightforward to obtain
\begin{align}
\label{eq:EvexpansionTW}
E^{\text{TW}}(v) = e_{\rm P} + v^2 \left( \frac 14 + \alpha^4 \| \partial_1 \varphi_{\rm P} \|_2^2  \right) + o( v^2), 
\end{align}
which agrees with \eqref{eq:Ev-expansion-alpha-0} for the case $\alpha =0$, and also with \eqref{eq:thm} to leading order in $v$. In particular, \eqref{eq:mass-predTW} holds.

\subsection{Effective mass with alternative definition for the electron's velocity} 
\label{sec:diff-vel}
In this Section, we discuss a different approach to the definition of the effective mass. This approach is based on an alternative way of defining the electron's position and velocity. While in Section \ref{sec:vel} we use the standard definition from quantum mechanics,  here we use a definition similar to the one of the phonons' position and velocity \eqref{eq:defXph}.  For this purpose, we recall Remark \ref{rem:E} and that $\delta^*$ has been chosen such that the condition $\E (\psi_0) \leq \G( \psi_0, \varphi_0) \leq e_{\rm P} + \delta^*$ ensures that for all times there exists a unique $L^2$-projection $e^{i\theta(t)}\psi_{\rm P}^{y(t)}$ of  $\psi_t$ onto the manifold $\mathcal{M}_\E$.  Then, we define the electron's position and velocity by
\begin{align}
\label{eq:def-tildeX}
\widetilde{\text{X}}_{\text{el}}(t) = y(t), \quad \widetilde{\text{V}}_{\text{el}}(t) = \dot{y}(t).
\end{align}
Similarly to the conditions (i) and (ii) in Section \ref{sec:vel}, we define the set of admissible initial data as  
\begin{align}
\widetilde{I}_v = \lbrace (\psi_0, \varphi_0) \, \vert \,  \text{(i),(ii') are satisfied} \rbrace
\end{align}
where 
\begin{itemize}
	\item[(ii')] $\widetilde{\text{V}}_{\text{el}}(0) =\Vph(0)=v (1,0,0)$. 
\end{itemize}

Note that we are leaving the parameter $\dot{\theta}(0)$ free, which in this case is also relevant. In other words, we have 
\begin{align}
\widetilde{I}_v = \cup_{\kappa\in \R} \widetilde{I}_{v,\kappa},
\end{align} 
where
\begin{align}
\widetilde{I}_{v,\kappa}=\lbrace (\psi_0, \varphi_0) \, \vert \,  \text{(i),(ii') are satisfied and}\; \dot{\theta}(0)=\kappa \rbrace.
\end{align}
Minimizing now the energy over all states of the set $\widetilde{I}_v$
\begin{align}
\widetilde{E}(v):=\inf_{(\psi_0,\varphi_0)\in \widetilde{I}_v} \G(\psi_0,\varphi_0),
\end{align}
leads to an energy expansion in $v$ that differs from the one of Theorem \ref{thm:freephase} in its second order. 

\begin{prop}
	\label{rmk:EV}
	As $v \rightarrow 0$,  we have 
	\begin{align}
	\label{eq:Evtilde}
	\widetilde{E}(v) =e_{\rm P}+v^2 \left(  \frac{\| \nabla \psi_{\rm P} \|_2^4}{3 \|\nabla \varphi_{\rm P} \|_2^2}  + \frac{\alpha^4}{3} \| \nabla \varphi_{\rm P} \|_2^2  \right) + O(v^3).	
	\end{align}
\end{prop}

The energy expansion in \eqref{eq:Evtilde} leads to the effective mass 
\begin{align}
 \widetilde{m}_{\text{eff}} = \lim_{v \rightarrow 0} \frac{\widetilde{E}(v) -e_{\rm P}}{v^2/2} = \frac{2 \| \nabla \psi_{\rm P} \|_2^4}{3 \|\nabla \varphi_{\rm P} \|_2^2 } + \frac{2\alpha^4}{3} \| \nabla \varphi_{\rm P} \|_2^2 
\end{align} 
which agrees with \eqref{eq:OurEffMass} and \eqref{eq:mass-predTW} in leading order for large $\alpha$ only (and thus still confirms the Landau--Pekar prediction \cite{landau1948effective}), but differs in the $O(1)$ term.  In fact, it turns out that $\widetilde{m}_{\text{eff}} < m_{\text{eff}}$ with $m_{\text{eff}}$ defined in \eqref{eq:mass-predTW}.   

This follows from the observation that the trial state
\begin{align}
\label{eq:trialstate-mass}
(\widetilde{\psi}_0,\widetilde{\varphi}_0)=\left(\frac{f_v \psi_{\rm P}+ivH_{\rm P}^{-1}\partial_1 \psi_{\rm P}}{\|f_v \psi_{\rm P}+ivH_{\rm P}^{-1}\partial_1 \psi_{\rm P}\|}, \varphi_{\rm P}+iv\alpha^2 \partial_1 \varphi_{\rm P}\right)
\end{align}
with $f_v = \frac{1}{2}\left( 1 + \sqrt{1 - v^2/(4 \| \partial_1 \psi_{\rm P} \|_2^2)} \right)$ (which coincides up to terms of order $v^2$ with the trial state \eqref{eq:trialstate}) is an element of $\widetilde{I}_{v,\bar \kappa}$ for $\bar \kappa= - \mu_{\rm P}+4\|\partial_1 \psi_{\rm P}\|_2^2(f_v-1)$ and is such that $\G (\widetilde{\psi}_0,\widetilde{\varphi}_0) = e_{\rm P} + m_{\text{eff}}\;  v^2/2 + O(v^3)$.  Thus,  $\widetilde{m}_{\text{eff}} \leq m_{\text{eff}}$ and equality holds if and only if equality (up to terms $o(v^2)$) holds in \eqref{eq:completing-square2}. This is the case if and only if 
\begin{align} 
Q_{\psi_{\rm P}} \left( \Im\widetilde{\psi}_0 - cv  \partial_1 \psi_{\rm P}\right) = o(v).
\end{align}
Using \eqref{eq:trialstate-mass}, equality holds if and only if
\begin{align} 
 0 = H_{\rm P}^{-1}\partial_1 \psi_P -  c  \partial_1 \psi_{\rm P} = -\left( x_1/2 +  c  \partial_1 \right)  \psi_{\rm P} 
\end{align}
i.e. , recalling the radiality of $\psi_{\rm P}$,  if and only if $\psi_{\rm P}$ is a Gaussian with variance $\sigma^2 = 1/(2c)$. Since $\psi_{\rm P}$ satisfies the Euler--Lagrange equation
\begin{align}
H_{\rm P} \psi_{\rm P}=0 \iff V_{\varphi_{\rm P}} \psi_{\rm P}=(-\Delta+\mu_P) \psi_{\rm P},
\end{align}
it cannot be a Gaussian and therefore $\tilde{m}_{\text{eff}}<m_{\text{eff}}$.

We present only a sketch of the proof of Proposition \ref{rmk:EV}, since it uses very similar arguments as the proof of Theorem \ref{thm:freephase}.  

\begin{proof}[Sketch of Proof of Proposition \ref{rmk:EV}]  \textbf{Upper bound}: We use the alternative trial state 
\begin{align}
(\widetilde{\psi}_0,\widetilde{\varphi}_0)=\left(\frac{f_v \psi_{\rm P}+ivc\partial_1 \psi_{\rm P}}{\|f_v \psi_{\rm P}+ivc\partial_1 \psi_{\rm P}\|}, \varphi_{\rm P}+iv\alpha^2 \partial_1 \varphi_{\rm P}\right),
\end{align}
with 
\begin{align}
f_v:=\frac{1+\sqrt{1+4c^2v^2\|\partial_1 \psi_{\rm P}\|^2}}2, \quad c:=\frac{\|\partial_1 \psi_{\rm P}\|^2}{\|\partial_1\varphi_{\rm P}\|^2} . 
\end{align}
With similar arguments as in the previous section, one can verify that $(\widetilde{\psi}_0,\widetilde{\varphi}_0)\in \widetilde{ I}_v$, in particular  $(\widetilde{\psi}_0,\widetilde{\varphi}_0)\in \widetilde{ I}_{v, \kappa}$ with $\kappa = -\mu_{\rm P}+\frac{-1+\sqrt{1+4 c^2v^2\|\partial_1 \psi_{\rm P}\|^2}}{2c}$.

Note that, similarly to \eqref{eq:defVph2}, one can derive necessary conditions for the velocities  $\dot{y}(0), \dot{\theta}(0)$ (using $\widetilde{\text{X}}_{\text{el}} (0) =0, \theta (0) =0$),  namely 
\begin{align}
\label{eq:necessary1}
&\bra{[h_{\Re\widetilde{\varphi}_0}+\dot{\theta}(0)] \Im \widetilde{\psi}_0- \dot{y}(0) \cdot \nabla  \Re \widetilde{\psi}_0}\ket{(u\cdot \nabla) \psi_{\rm P}}=0 \quad \text{for all} \quad u \in \mathbb{S}^2 
\end{align}
and 
\begin{align}
\label{eq:necessary2}
&\bra{\psi_{\rm P}}\ket{(h_{\Re\widetilde{ \varphi}_0}+\dot{\theta}(0)) \Re \widetilde{\psi}_0+\dot{y}(0) \cdot \nabla \Im \widetilde{\psi}_0}=0  .
\end{align}
Straightforward computations then show that
\begin{align}
\widetilde{E}(v)\leq \G(\widetilde{\psi}_0,\widetilde{\varphi}_0)=e_{\rm P}+v^2 \left(\frac{\| \partial_1 \psi_{\rm P} \|_2^4}{\|\partial_1 \varphi_{\rm P} \|_2^2} +  \alpha^4 \| \partial_1 \varphi_{\rm P} \|_2^2  \right) + O(v^3).
\end{align}

\textbf{Lower bound}: We proceed similarly to the lower bound in the previous section.  First, we assume w.l.o.g.  that
\begin{align}
P_{L^2}^{\minE} ( \psi_0 ) = \psi_{\rm P}^{y(0)}, \quad P_{L^2}^{\minF} ( \varphi_0 ) = \varphi_{\rm P} ,
\end{align}
i.e.,  centering with respect to translations and changes of phase. 
We can then substitute the two conditions of (ii') and the conditions for $\psi_{\rm P}^{y(0)} $ (resp.  $\varphi_{\rm P}$)  to be the $L^2$-projection of $\psi_0$ (resp. $\varphi_0$) onto $\mathcal{M}_\E$ (resp.  $\mathcal{M}_\F$) with their analogue necessary conditions (whose computations proceed along the lines of \eqref{eq:necessary1} and \eqref{eq:necessary2}).  With this discussion, we are left with the task of minimizing $\G$ over the set
\begin{align}
\widetilde{I}_v':=\bigcup_{\kappa \in \R} \widetilde{I}'_{v,\kappa}
\end{align}
with 
\begin{align}
\label{eq:def-tildeI}
\widetilde{I}'_{v,\kappa}:=\Big\{(\psi_0,\varphi_0) \in\widetilde{ I}^*\,|
&P_{\nabla \psi_{\rm P}^{y(0)}}\left[(h_{\Re \varphi_0}+\kappa) \Im \psi_0 - v\partial_1 \Re \psi_0\right]=0 ,\nonumber\\
&P_{\psi_{\rm P}^{y(0)}}\left[(h_{\Re \varphi_0}+\kappa) \Re \psi_0+v\partial_1 \Im \psi_0 \right]=0,\nonumber\\
&P_{\nabla \varphi_{\rm P}} (\Im \varphi_0-v\alpha^2 \partial_1 \Re\varphi_0)=0\Big\}.
\end{align}
and 
\begin{align}
\widetilde{I}^*:=\left\{(\psi_0,\varphi_0) \,|\, \G(\psi_0,\varphi_0)\leq e_{\rm P}+\delta^*,\,\,\|\psi_0\|_2^2=1, \, \Re \psi_0 \perp \nabla \psi_{\rm P}^{y(0)},\, \Re \varphi_0 \perp \nabla \varphi_{\rm P}\right\}.
\end{align}

As in the previous section, one can argue by coercivity of $\E$ and $\F$ and the upper bound that it is possible to restrict to initial conditions such that $ \| \delta_2 \|_2,  \| \delta_1 \|_{H^1}, y(0) $ are all $O(v)$.  Moreover, the second constraint of the r.h.s.  of \eqref{eq:def-tildeI} shows that $\kappa = -\mu_{\rm P} + O(v)$. Thus, we are left with minimizing $\G$ over the set 
\begin{align}
\label{eq:tilde-I''}
\widetilde{I}''_v:=\widetilde{I}_v'\cap \left\{\kappa+\mu_{\rm P} =O(v),\, \|\delta_1\|_{H^1}=O(v),\,\|\delta_2\|_2=O(v)\right\}.
\end{align}
The lower bound is proven in the same way as before.  But instead of the constraint \eqref{eq:constranit-Im}, this time we need to minimize w.r.t. 
\begin{align}
\label{eq:constraint7}
P_{\nabla \psi_{\rm P}^{y(0)}}\left[(h_{\Re \varphi_0}+\kappa) \Im \psi_0 - v\partial_1 \Re \psi_0\right]=0. 
\end{align}
Since  $\kappa+\mu_{\rm P}, y_0, \|\delta_1\|_{H^1}$ and $\|\delta_2\|_2$ are all order $v$ and  $\psi_{\rm P} \in C_0^\infty (\R^3)$ (and these facts also allow to infer that $\psi_{\rm P}^{y(0)} = \psi_{\rm P} + O(v)$), the constraint \eqref{eq:constraint7} can be written as 
\begin{align}
\langle \nabla \psi_{\rm P} \vert H_{\rm P} \vert \Im \psi_0 \rangle = v \| \partial_1 \psi_P \|_2^2 (1,0,0) + O(v^2).
\end{align}
Denoting $c= \| \partial_1 \psi_P \|_2^2/\| \partial_1 \varphi_P \|_2^2$, we complete the square 
\begin{align}
\expval{H_{\rm P}}{\Im\psi_0} & = \expval{H_{\rm P}}{\Im\psi_0 - v c  \partial_1 \psi_{\rm P}} \notag \\
& \quad + 2v c \bra{\Im\psi_0}\ket{H_{\rm P} \vert \partial_1  \psi_{\rm P}} - c^2  v^2  \expval{H_{\rm P}}{  \partial_1 \psi_{\rm P}}\notag \\
&  \geq  2 c v \bra{\Im\psi_0}\ket{ H_{\rm P} \partial_1  \psi_{\rm P}} -c^2  v^2 \expval{H_{\rm P}}{  \partial_1 \psi_{\rm P}}  .  \label{eq:completing-square2}
\end{align}
 With the constraint \eqref{eq:constraint7} and  $\langle \partial_i \psi_{\rm P}\vert H_{\rm P} \vert \partial_j \psi_{\rm P} \rangle = \delta_{i,j}\| \partial_j \varphi_{\rm P} \|_2^2$, we arrive at \eqref{eq:Evtilde}. 
\end{proof}

\section{Conclusions} 
While a rigorous determination of the effective mass of a polaron described by the Fr\"ohlich model remains an outstanding open problem, we solve here the classical analog of this problem, where the polaron is described by the Landau--Pekar equations. Even though these equations are often invoked in heuristic derivations of the effective polaron mass, it is not at all obvious how to make such derivations rigorous since they rely, in one form or another, on the assumption of the existence of traveling waves. As argued above, the latter can not be expected to exist, however. We overcome this problem by introducing a novel variational principle, minimizing the Pekar energy functional over states of given initial velocity $v$, which can be defined in a natural way for all low-energy states. We hope that this novel point of view may in the future also shed  some light on the corresponding problem for the Fr\"ohlich polaron, in particular in view of the recent derivation \cite{LMRSS} of the Landau--Pekar equations from the Fr\"ohlich model in the strong coupling limit. 

\appendix
\section{Well-posedness and regularity of the projections onto $\minF$}
\label{app:projection}
 
Similar arguments to the ones used in the following proof are contained in \cite{feliciangeli2021strongly}, where  the functional $\mathcal{F}$ is investigated in the case of a torus in place of $\R^3$. Remark \ref{rem:E} on the properties $\minE$ can be shown with a similar approach, but we omit its proof. 

\begin{proof}[Proof of Lemma \ref{lem:minEFproj}] We need to prove that there exists $\delta>0$ such that for any $\varphi\in (\minF)_\delta$ there exists a unique $z_\varphi$ identifying the projection of $\varphi$ onto $\minF$, and such that $z_{\varphi}$ is differentiable at any $\varphi\in (\minF)_{\delta}$.  As the problem is invariant w.r.t. translations,  we can w.l.o.g. restrict to show differentiability at $\varphi_0 \in (\minF)_{\delta}$ such that $z_{\varphi_0}=0$.

We define the function $F: L^2 ( \Rtre ) \times \Rtre \rightarrow \Rtre$ given, component-wise, by 
\begin{align}
 F_i (\varphi, z) = \Re \langle \varphi \vert \partial_i \varphi_P^z \rangle \quad \text{for} \quad i=1,2,3.
\end{align}
By definition of $z_\varphi$, we have $F( \varphi_0, 0 ) = 0$ and $F( \varphi, z_{\varphi} ) = 0$, for any $\varphi$ in a sufficiently small neighborhood of $\varphi_0$. Hence, we set out to use the implicit function theorem to determine properties of $z_{\varphi}$. Observe that,  for any  $\eta \in L^2 ( \Rtre)$, $z \in \Rtre$ and $i,j\in \{ 1,2,3 \} $,  we have 
\begin{align}
\partial_t  F_i ( \varphi + t \eta, z ) 
= \Re \langle \eta \vert \partial_i \varphi_P^z \rangle \quad \text{and} \quad  \partial_{z_j} F_i ( \varphi, z) = - \Re \langle \varphi \vert \partial_i\partial_j \varphi_P^z \rangle \, .
\end{align}
Since $\varphi_P \in C^\infty ( \Rtre)$, the map $(\minF)_{\delta} \ni \varphi \mapsto \text{det} \left( \frac{\partial F_i}{\partial z_j }  (\varphi, z)  \right)_{i,j=1, \dots, 3}$ is continuous w.r.t the $L^2$-norm and, by radiality of $\varphi_P$, 
\begin{align}
\text{det} \left( \frac{\partial F_i}{\partial z_j } (\varphi_P, 0) \right)_{i,j=1, \dots, 3} =  \frac{1}{9} \| \nabla \varphi_P \|_2^2>0\,.
\end{align}
Thus,  it follows that $\text{det} \left( \frac{\partial F_i}{\partial z_j }(\varphi_0, 0)   \right)_{i,j=1, \dots, 3} > 0 $, uniformly in $\varphi_0$ for sufficiently small $\delta >0$.  By the implicit function theorem,  there exists a unique differentiable  $z_\varphi: (\minF)_{\delta} \rightarrow \Rtre $ whose partial derivative in the direction $\eta \in L^2 (\Rtre)$ at $\varphi_0$  is given by 
\begin{align}
\partial_t   z_{\varphi_0 + t \eta} \restriction_{t=0} = \left[\left( \frac{\partial F_i}{\partial z_j }(\varphi_0, z_{\varphi_0})   \right)_{i,j=1, \dots, 3}\right]^{-1}\Re  \langle \eta \vert \partial_i \varphi_P^{z_{\varphi_0}} \rangle \; . 
\end{align}

\end{proof}

{\bf Acknowledgments:} We thank Herbert Spohn for helpful comments. Funding from the European Union's Horizon 2020 research and innovation programme under the ERC grant agreement No. 694227 (D.F. and R.S.) and under the Marie Sk\l{}odowska-Curie  Grant Agreement No. 754411 (S.R.) is gratefully acknowledged.

\bibliographystyle{siam}

\end{document}